\documentclass[11pt]{article}

\usepackage{fullpage}

\usepackage{amsmath,amssymb,mathtools,xcolor}
\usepackage{enumitem}

\usepackage{amsthm,thmtools}
\declaretheoremstyle[bodyfont=\it]{it}
\declaretheoremstyle[qed=$\lrcorner$,bodyfont=\rm]{rm}

\declaretheorem[style=it]{theorem}

\declaretheorem[style=it,name=Lemma,numberlike=theorem]{lemmalpha}
\declaretheorem[style=it,name=Corollary,numberlike=theorem]{corollaryalpha}

\declaretheorem[style=it]{lemma}
\declaretheorem[style=rm,numberlike=lemma]{remark}
\declaretheorem[style=it,numberlike=lemma]{claim}

\newcommand{\pr}{\mathrm{Pr}}

\usepackage{bm}
\def\st{\;{\bm :}\;}
\def\set#1{\left\{#1\right\}}
    
\let\tmp=\phi \let\phi=\varphi \let\varphi=\tmp
\let\tmp=\epsilon \let\epsilon=\varepsilon \let\varepsilon=\tmp

\makeatletter

\usepackage{accents}

\def\mkmathletter#1#2{%
    \expandafter\gdef\csname#1#2\endcsname%
           {\ensuremath{\csname math#2\endcsname{#1}}}%
}
\def\mkmathletters#1#2{%
        \@for\ltr:={#2}\do{\expandafter\mkmathletter\ltr{#1}}%
}
\edef\letters{a,b,c,d,e,f,g,h,i,j,k,l,m,n,o,p,q,r,s,t,u,v,w,x,y,z}
\edef\Letters{A,B,C,D,E,F,G,H,I,J,K,L,M,N,O,P,Q,R,S,T,U,V,W,X,Y,Z}
\edef\Lletters{\letters,\Letters}
\makeatother
\mkmathletters{bb}{\Letters}    
\mkmathletters{cal}{\Lletters}   
\mkmathletters{bf}{\Lletters}   
\mkmathletters{tilde}{\Lletters}
\mkmathletters{bar}{\Lletters}  

\let\dmo=\DeclareMathOperator
\dmo\AK{\mathsf{AK}}
\dmo\Copy{\mathsf{Copy}}
\dmo\dist{\mathsf{dist}}

\usepackage{graphicx}
\usepackage{wrapfig}
\usepackage{tikz}
\usetikzlibrary{shapes,backgrounds}
\usetikzlibrary{decorations.pathreplacing}
\usetikzlibrary{arrows,decorations.pathmorphing,backgrounds,fit,positioning,shapes.symbols,chains}
\usetikzlibrary{arrows}

\definecolor{lightred}{rgb}{1, .30, 0.30}
\definecolor{lightblue}{rgb}{0.30, .30, 1.0}

\title{Structural Properties of Entropic Vectors and Stability of the Ingleton Inequality}

\author{Rostislav  Matveev and Andrei Romashchenko}

\begin{document}

\maketitle

\abstract{ We study constrained versions of the Ingleton inequality in
  the entropic setting and quantify its stability under small
  violations of conditional independence. Although the classical
  Ingleton inequality fails for general entropy profiles, it is known
  to hold under certain exact independence constraints. We focus on
  the regime where selected conditional mutual information terms are
  small (but not zero), and the inequality continues to hold up to
  controlled error terms.
  A central technical tool is a structural lemma that “materializes’’ part of
  the mutual information between two random variables, 
  implicitly capturing the effect of infinitely many non-Shannon--type
  inequalities. This leads to conceptually transparent proofs without explicitly
  invoking such infinite families.
  Some of our bounds recover, in a unified way, what can also be deduced from
  the infinite families of inequalities of Mat\'u\v{s} (2007) and of
  Dougherty--Freiling--Zeger (2011), while others appear to be new.
}

\section{Introduction}\label{s:intro}
In this paper we study inequalities for Shannon entropy, in particular
variants of the Ingleton inequality in the entropic setting.  At a
high level, the goals of this paper can be summarized as two broad
objectives:
\begin{enumerate}[label=(\Roman*)]
\item gaining a better understanding of the geometry of the boundary
  of the cone of almost-entropic vectors;
\label{obj:i}
\item developing tools for further progress in understanding
  conditional independence structures.
\label{obj:ii}
\end{enumerate}
More specifically, we focus on conditional versions of the Ingleton
inequality and on related structural properties of entropy.  It is
well known that the Ingleton inequality in its original form does not
hold for the entropies of certain collections of jointly distributed
random variables; see,
e.g.~\cite{matuvs1995conditional,csirmaz1996dealer,zhang1998characterization,matu1999conditional,hammer2000inequalities}.
However, under certain restrictions (or with suitable additional
terms) the inequality may still hold. We study several manifestations
of this phenomenon.

\subsection{The cone of almost-entropic points}
We begin with a brief review of the standard framework for information
inequalities.  For a tuple of random variables
$(X_i\st i=0,\dots,n-1)$ its \emph{entropy profile} is the function on the
power set of the set $\set{0,\dots,n-1}$ indexing the tuple
whose value on a set $J\in 2^{[n]}$ is the entropy of the joint
distribution $(X_{i}\st i\in J)$. Since the entropy of the empty
collection is always zero, we drop the corresponding value from
consideration. Thus, every entropy profile is effectively a point in
$\Rbb^{2^{n}-1}$.  In particular, for a four-tuple of random variables,
 its entropy profile is a 15-dimensional vector.

A point in $\Rbb^{2^{n}-1}$ is called \emph{entropic} if it is the
entropy profile of the probability distribution of some tuple
$(X_i\st i=0,\dots,n-1)$.  The set of all entropy profiles for all
$n$-tuples of random variable is usually denoted $\Gamma_{n}^\star$.
The points in the closure of the set of entropy profiles are called
\emph{almost-entropic points}.  This set is denoted by
$\bar \Gamma_{n}^\star$.  It has been shown in~\cite{zhang1997non} that this
set is in fact a convex cone.  The geometry of this cone is quite
non-trivial.  
An exact description of this cone for $n\ge 4$ is still unknown.
However, there are some known inner and outer bounds for this set.
In what follows we briefly discuss the basic facts about 
$\bar \Gamma_{4}^\star$.

An entropy profile of a four-tuple of jointly distributed random
variables can be thought of as a function assigning a real number to
every non-empty subset of $\{0,1,2,3\}$.  It is known that every such
function $\phi$ necessarily satisfies the inequalities expressing the
properties of \emph{non-negativity}, \emph{monotonicity}, and
\emph{submodularity}:
\begin{align*}
  0\leq \phi(A) \le \phi(B)
  &&&
     \text{for all sets}\ A\subset B \subset \set{0,1,2,3},
  \\
  \phi(A\cup B) + \phi(A\cap B)  \le \phi(A) + \phi(B)
  &&&
      \text{for all sets}\ A, B \subset \set{0,1,2,3}.
\end{align*}
In other words, every entropy profile is a non-negative, monotone,
submodular function. In information theory, these constraints for $\phi$ are
usually called \emph{Shannon inequalities}.  In algebra, the functions
that respect these constraints are usually called \emph{rank functions of a polymatroid}.
This is a fairly standard and well-studied class of objects. 

The set of all $\phi$ satisfying these inequalities is traditionally
denoted $\Gamma_{4}$ (the cone of submodular functions on a four-point
set). This set is clearly a superset (i.e., an outer bound) of
$\bar \Gamma_{4}^\star$.
In other words, the axioms of polymatroids (\emph{non-negativity}, \emph{monotonicity}, and \emph{submodularity}) are necessary conditions for a point to be entropic.

The cone $\bar \Gamma_{4}^\star$ is strictly smaller than $\Gamma_{4}$; that is, besides the inequalities valid for all polymatroids, 
there exist other linear inequalities that are valid for all almost-entropic points (for distributions of $n\ge 4$ random variables).
The inequalities that are valid for all almost-entropic points but do not follow from the Shannon inequalities are called \emph{non-Shannon type inequalities}.
The first example of such an inequality was discovered by Z.~Zhang and R.\,W.~Yeung in \cite{zhang1998characterization}.
In the following years many other non-Shannon type inequalities were discovered; see, e.g., \cite{makarychev2002new,dougherty2006six,xu2008projection,dougherty2011non,csirmaz2016using}.
Moreover, F.~Mat\'{u}\v{s} proved in \cite{matus2007infinitely} that there exist infinitely many independent non-Shannon type inequalities 
(in other words, $\bar \Gamma_{4}^\star$ is not polyhedral, as well as all $\bar \Gamma_{n}^\star$ for $n>4$).

In the early years after the discovery of the first non-Shannon type inequalities, much effort was devoted to understanding how large the gap between $\bar \Gamma_{4}^\star$ and $\Gamma_{4}$ is. 
Accordingly, particular interest was placed on non-Shannon inequalities that exclude a significant region of the cone $\Gamma_{4}$.
In recent years, attention has shifted to somewhat different questions. For example: what is the geometry of the boundary of the cone $\Gamma_{4}$? 
How accurately can it be approximated by piecewise algebraic surfaces? How smoothly does the curved part of the boundary approach a given flat face?
The main results of this paper are primarily relevant to the second type of questions, namely those related to the geometry of the boundary of $\bar \Gamma_{4}^\star$
(see Objective~\ref{obj:i} on p.~\pageref{obj:i}).

The main results of this paper are centered around the \emph{Ingleton inequality} described in \cite{ingleton}.
To discuss this inequality, we introduce the following notation: 
for a set function $\phi$ define
\begin{align*}
  \phi(x:y)
  &:=
    \phi(\set{x})+\phi(\set{y})-\phi(\set{x,y}),\\ 
  \phi(x:y|z)
  &:=
    \phi(\set{x,z})+\phi(\set{y,z})-\phi(\set{x,y,z}) - \phi(\set{z}).
\end{align*}
Now the Ingleton inequality can be written in the following compact form:
\begin{equation}\label{eq:ingleton1234}
  \phi(0:1) \leq \phi(0:1|2) + \phi(0:1|3) + \phi(2:3).
\end{equation}
Permuting the set $\set{0,1,2,3}$ we obtain $\binom{4}{2} = 6$ similar
inequalities that differ in the order of the variables.  We denote by
$\Gamma^{\rm Ing}_4$ the cone of all non-negative, monotone,
submodular functions $\phi$ on $\set{0,1,2,3}$ that satisfy, in
addition to the Shannon inequalities, the six versions of the Ingleton
inequality on a four-point set
(this cone under a different notation was discussed, e.g., in \cite{chan2010existence,weilenmann2017analysing}).

The Ingleton inequality is a \emph{sufficient} condition for being an entropic point, i.e., 
every $\phi$ in $\Gamma^{\rm Ing}_4$ is entropic, see
\cite{matus-studeny1995conditional} (independently, this was proved in
\cite{hammer2000inequalities}).
In other words, $\Gamma^{\rm Ing}_4$ is a subset (inner bound) of
the set of almost-entropic points.
It is well known that the Ingleton inequality is not a \emph{necessary} condition for being an entropic point, i.e., 
there exist entropic $\phi$ outside  $\Gamma^{\rm Ing}_4$,
see \cite{matuvs1995conditional,csirmaz1996dealer,zhang1998characterization,matu1999conditional,hammer2000inequalities,boston2020violations,matus2016entropy}.

The sets $\Gamma_4$ and $\Gamma_4^{\mathrm{Ing}}$ are both polyhedral
cones, and the set of almost-entropic points is sandwiched between the
two:
\( \Gamma_{4}^{\mathrm{Ing}}\subsetneq\bar
\Gamma_4^\star\subsetneq\Gamma_{4}.  \) The closure of
$\Gamma_4\setminus \Gamma_{4}^{\mathrm{Ing}}$ consists of six
closed convex cones with disjoint interiors,  \cite{matveev2020tropical}.   
They are symmetric under permutations of the set $\set{0,1,2,3}$.  
From now on, we consider one of these cones in which the strict inequality~\eqref{eq:ingleton1234} is violated 
and, following \cite{matveev2020tropical}, refer to it as the \emph{non-Ingleton cone}.
Let us briefly recall how this object can be viewed in geometric terms (see \cite{matveev2020tropical} or \cite{chan2010existence,guille2011minimal} for more details).
The non-Ingleton cone has a 14-dimensional simplex as its base.  One
of the 13-dimensional faces of this simplex is defined by the Ingleton
equality and is a common face with the base of
$ \Gamma_{4}^{\mathrm{Ing}}$. The remaining vertex of the simplex is
the so-called \emph{special point}. It is neither an entropic nor an almost
entropic point, as was shown in~\cite{zhang1998characterization}.  

The study of the cone of almost-entropic points essentially reduces to the
question of what neighborhood around the special point is free of almost
entropic points, or, equivalently, what neighborhood of the face shared
between the Ingleton and non-Ingleton regions is filled with almost
entropic points.%
\footnote{ If we denote by $\alpha_{1},\dots,\alpha_{15}$ the convex
  coordinates in the non-Ingleton simplex (the base simplex of
  non-Ingleton cone), such that $\alpha_{15}=1$ is the defining equation of the
  special point, while $\alpha_{15}=0$ is the equation of the face
  common with the Ingleton-region, then the almost -entropic region is
  described by the inequality
\[
  \alpha_{15}\leq \Phi(\alpha_{1},\dots,\alpha_{14})
\]
where $\Phi$ is some concave function on the $13$-dimensional
simplex. Ultimately, the complete characterization of
$\bar{\Gamma}^{*}_{4}$ reduces to the description of the function $\Phi$.
The main results of this paper contribute to a better understanding
of the behavior of $\Phi$ near the boundary of its domain.} %
Thus, the Ingleton inequality plays a crucial role in the geometry of $\bar{\Gamma}^{*}_{4}$. 
On the one hand, this inequality together with all its permutations provides an inner bound for this cone. 
On the other hand, many non-Shannon type inequalities  can be written in the form of the Ingleton inequality with a relatively small additional ``correction term'', 
which give an outer bound for this cone, see the discussion in  \cite[section~2]{weilenmann2017analysing}.
Systematic investigations of the geometry of the cones $\Gamma_{4}$, $\bar \Gamma_4^\star$, and $\Gamma_{4}^{\mathrm{Ing}}$ 
can be found in~\cite{csirmaz2025exploring,matus2016entropy,matveev2020tropical}.
In general, the complete description of $\bar \Gamma_4^\star$ remains far from being understood,
despite considerable efforts devoted to this problem.

\subsection{Constraint versions of the Ingleton inequality.}
In what follows we consider the Ingleton inequality in the entropic setting and discuss several versions of this inequality 
obtained by imposing constraints and/or adding a negligible correction term, that are valid for entropic and almost-entropic points.
To avoid ambiguity and improve the readability of
the notation we denote the four-tuple of jointly distributed random
variables $(X,Y,A,B)$, and rewrite~\eqref{eq:ingleton1234} with the
usual notation for the conditional mutual information as
\begin{equation}\label{eq:ingleton}
\tag{ING} 
  I(X:Y) \leq I(X:Y|A) + I(X:Y|B) + I(A:B)
\end{equation}
When discussing constraint inequalities, we assume that some of the quantities
\begin{equation}\label{eq:assumptions}
  I(X:A|Y), I(Y:A|X), I(X:B|Y) , I(Y:B|X), I(X:Y|A),  I(X:Y|B)
\end{equation}
are small compared to the other information quantities.
It is known that~\eqref{eq:ingleton} holds for joint distributions
$(X,Y,A,B)$ satisfying certain independence constraints, that is,
whenever some of the quantities in~\eqref{eq:assumptions} vanish,
see  \cite{kaced2013conditional,studeny2021conditional}.
We will focus on the following two implications:
\begin{align}
\label{ineq:cond-ineq-1}
I(X:A|Y) = I(Y:A|X) = 0 \;&\Longrightarrow\;  \eqref{eq:ingleton},\\
\label{ineq:cond-ineq-2}
I(X:Y|A) = I(X:A|Y) = 0 \;&\Longrightarrow\;  \eqref{eq:ingleton}.
\end{align}
To the best of our knowledge, the conditional inequalities \eqref{ineq:cond-ineq-1} and \eqref{ineq:cond-ineq-2} were first explicitly formulated in this form in \cite{kaced2013conditional}. 
The proof given in \cite{kaced2013conditional} relied on a deep and nontrivial result of Mat\'u\v{s}, namely the infinite series of conditional inequalities established in \cite{matus2007infinitely}.
Studen\'y proposed in \cite{studeny2021conditional} a different proof of \eqref{ineq:cond-ineq-1} and \eqref{ineq:cond-ineq-2}, which is more intuitive and technically simpler. 
His argument is based on the following simple lemma. %
\begin{lemma}[see lemma~8 in \cite{studeny2021conditional}]
\label{l:studeny}
Every discrete distribution of random variables $(X,Y,A)$ satisfying
\[
I(X:A|Y)=0,\ I(Y:A|X)=0
\]
can be extended by a discrete random variable $W$ such that in the joint distribution $(X,Y,A,W)$
the variable $W$ is functionally dependent on both $X$ and $Y$, and moreover
\[
I(X,Y:A|W)=0.
\]
\end{lemma}
The proof of \eqref{ineq:cond-ineq-1} and \eqref{ineq:cond-ineq-2} then proceeds by applying this lemma in combination with standard Shannon inequalities. 
Conceptually, this argument is reminiscent of the classical proof of the Ingleton inequality for linearly representable matroids.
\begin{remark}
Lemma~\ref{l:studeny} and its use in proving information inequalities are analogous to the Ahlswede--K\"orner lemma and its role in the derivation of non-Shannon-type information inequalities in \cite{makarychev2002new}.
Recall that the Ahlswede--K\"orner lemma was used in  \cite{farras2018improving} in computer-assisted proofs of lower bounds for problems of secret sharing.
We believe that Lemma~\ref{l:studeny} (and its ``smooth'' versions discussed below) can be also used in computer-assisted proofs of conditional information inequalities.
\end{remark}

However, the simplicity of Studen\'y's approach comes at a price. 
Lemma~\ref{l:studeny} was proved in \cite{studeny2021conditional} for entropic points, but not for almost-entropic ones. 
Consequently, the resulting proof of \eqref{ineq:cond-ineq-1} and \eqref{ineq:cond-ineq-2} applies only to entropic points 
(whereas the more involved argument in \cite{kaced2013conditional} shows that these inequalities also hold for almost-entropic points).

In this paper we introduce a technical tool that combines the transparency of Studen\'y's approach with the power of Mat\'u\v{s}'s method.
We prove a robust version of Lemma~\ref{l:studeny} that applies to almost-entropic points, see Corollary~\ref{cor:W-bound} and Corollary~\ref{cor:W-bound-0}.
Moreover, we establish a trade-off between the error in the assumptions
\begin{equation}
\label{eq:approx-cond}
I(X:A|Y) \approx 0,\qquad I(Y:A|X)\approx 0
\end{equation}
and the possible deviation in the conclusions
\begin{equation}
\label{eq:approx-conclusion}
H(W|X)\approx 0,\qquad H(W|Y)\approx 0,\qquad I(XY:A|W)\approx 0.
\end{equation}
That is, we assume that the approximate equations in \eqref{eq:approx-cond} hold with precision $\epsilon$ and 
analyze the resulting asymptotics in \eqref{eq:approx-conclusion}.

Throughout the paper we formulate results of this type using asymptotic notation, without specifying the multiplicative constants hidden in the $O(\cdot)$ terms. 
We believe that refining the asymptotic behavior, rather than optimizing specific constants, will help reveal fundamental properties of the cone of almost-entropic points,
see also Remark~\ref{rem:3}.
For a discussion of non-asymptotic and unconditional versions of Theorems~\ref{th:sqrt}, \ref{th:sqrt-ii}, and~\ref{th:eloge}, see Remark~\ref{r:matus-conditional-ineq}.

Let us mention that \eqref{ineq:cond-ineq-1} and \eqref{ineq:cond-ineq-2} helped to classify the \emph{conditional independence structures}
over $4$ random variables,  \cite{studeny2021conditional,boege2024no} (cf.~Objective~\ref{obj:ii} on p.~\pageref{obj:ii}).
The classification of similar structures for $n>4$ remains an open problem, partly due to the rapidly increasing technical difficulties as the dimension grows.
We hope that the conceptual tools proposed in this paper, such as the robust version of Lemma~\ref{l:studeny}, may contribute to progress in this direction.

\subsection{Main technical results}
We focus on the possible violation of the Ingleton inequality in the
regime where some of the mutual information terms
in~\eqref{eq:assumptions} are relatively small but still positive.  We
consider three cases:
\begin{enumerate}[label=(\roman*)]
\item\label{i:xay,ayx}
  \(\max\{I(X:A|Y),\, I(Y:A|X)\} \le \epsilon\, I(X:Y)\);
\item\label{i:xay,xya}
  \(\max\{I(X:A|Y),\, I(X:Y|A)\} \le \epsilon\, I(Y:A)\)
\item\label{i:4bounds}
  \(\max\{I(X:A|Y),\, I(Y:A|X),\, I(X:B|Y),\, I(Y:B|X)\}
  \le \epsilon\, I(X:Y)\).
\end{enumerate}

If only two conditional mutual informations are bounded, as in
cases~\ref{i:xay,ayx} and~\ref{i:xay,xya}, then the Ingleton
inequality holds for such a four-tuple up to an additive error of order
$\sqrt{\epsilon}$. These cases are addressed in Theorem~\ref{th:sqrt} and
Theorem~\ref{th:sqrt-ii}. In contrast, when all four conditional mutual
informations are bounded, as in case~\ref{i:4bounds}, the Ingleton
inequality remains valid up to an additive error of order
$\epsilon \log \epsilon^{-1}$, as in Theorem~\ref{th:eloge} below.
\begin{theorem}[name=,restate=thsqrt]
\label{th:sqrt}
  Let $(X,Y,A,B)$ and $\epsilon\ge 0$ be such that
  \[
    \max\set{I(X:A|Y),I(Y:A|X)}\leq\epsilon\cdot I(X:Y).
  \]
  Then
  \[
    I(X:Y) \le I(X:Y|A) +  I(X:Y|B) + I(A:B) + O(\sqrt{\epsilon}) \cdot I(X:Y).
  \]
\end{theorem}

\begin{theorem}[name=,restate=thsqrtii]
  \label{th:sqrt-ii}
  Let $(X,Y,A,B)$ and $\epsilon\ge 0$ be such that
  \[
    \max\set{I(Y:A|X),I(X:Y|A)} \leq \epsilon\cdot I(X:A).
  \]
  Then
  \[
    I(X:Y)  \le  I(X:Y|A) +  I(X:Y|B) + I(A:B) + O(\sqrt{\epsilon})  \cdot I(X:A). 
  \]
\end{theorem}

\begin{theorem}[name=,restate=theloge]
  \label{th:eloge}
  Let $(X,Y,A,B)$ be a four-tuple of random variables such that
  \[
    \max\{I(X:A|Y), I(Y:A|X), I(X:B|Y), I(Y:B|X)\} \le \epsilon\cdot I(X:Y).
  \]
  Then
  \[
  I(X:Y) \leq I(X:Y|A) + I(X:Y|B) + I(A:B) +
  O(\epsilon\cdot\log\epsilon^{-1})\cdot I(X:Y). 
  \]
\end{theorem}
These three theorems capture certain geometric features of the boundary of the cone $\bar\Gamma^{\star}_4$, 
in particular on how smoothly the curved part of the boundary approaches the flat face defined by \eqref{eq:assumptions}.

\begin{remark}
\label{rem:3}
  A special case of Theorem~\ref{th:sqrt}, the
  implication~\eqref{ineq:cond-ineq-1}, is the \emph{essentially
    conditional inequality} (${\cal I}4'$) discussed
  in~\cite{kaced2013conditional}.  From~\cite[theorem~3
  (claim~4)]{kaced2013conditional} it follows that the term
  $O(\sqrt{\epsilon})$ in this theorem cannot be replaced by a
  function of order smaller than $O(\epsilon^{2/3})$.  Closing the
  remaining gap between the upper bound $O(\sqrt{\epsilon})$ and the
  lower bound $\Omega(\epsilon^{2/3})$ remains an open problem.
  
  Similarly, a special case of Theorem~\ref{th:sqrt-ii}, the
  implication~\eqref{ineq:cond-ineq-2}, is the \emph{essentially
    conditional inequality} denoted (${\cal I}5'$)
  in~\cite{kaced2013conditional}.  From~\cite[theorem~3
  (claim~5)]{kaced2013conditional} it follows that the bound in
  Theorem~\ref{th:sqrt-ii} is asymptotically tight, i.e., the term
  $O(\sqrt{\epsilon})$ in the conclusion of the theorem cannot be
  replaced by $o(\sqrt{\epsilon})$.
  
  We do not know whether the term $O(\epsilon\cdot\log\epsilon^{-1})$
  in Theorem~\ref{th:eloge} is asymptotically tight. This seems to be
  an important question, since if the bound is tight, it will imply
  that neither $\Gamma_{4}^{*}$ nor $\bar\Gamma_{4}^{*}$ are
  semi-algebraic sets, see~\cite{csirmaz2014structure,
    gomez2017defining}.
\end{remark}
\begin{remark}\label{r:matus-conditional-ineq}\ 
  Theorems~\ref{th:sqrt} and~\ref{th:sqrt-ii} can be derived from the
  infinite families of non-Shannon--type inequalities
  \begin{align*}
    I(X:Y)
    &\le I(X:Y|A) +  I(X:Y|B) + I(A:B)  \\
    &\quad
      {} + \frac1k I(X:Y|A) + \frac{k+1}2\big(I(X:A|Y) + I(Y:A|X)\big)
      \quad\text{for} \ k\in\Nbb
  \intertext{and}
    I(X:Y)
    &\le I(X:Y|A) +  I(X:Y|B) + I(A:B)  \\
    &\quad+
      \frac1k I(Y:A|X) + \frac{k+1}2\big(I(X:Y|A) + I(X:A|Y)\big)
      \quad\text{for } k\in\Nbb
  \end{align*}
  proven by Mat\'u\v s in~\cite{matus2007infinitely} (it suffices to
  choose a suitable value of $k$).  In fact, our proof of the lemma
  above runs parallel to the argument of Mat\'u\v s --- we apply the
  Copy Lemma and the Ahlswede--K\"orner lemma, while Mat\'u\v s used
  in~\cite{matus2007infinitely} a similar (and formally equivalent)
  method of \emph{adhesive extensions}.

  Likewise, Theorem~\ref{th:eloge} follows from the exponential
  family of inequalities in~\cite[Theorem 10]{dougherty2011non},
  \begin{align*}
    (2^{s-1} - 1)I(X:Y)
    &\leq
      2^{s-1}I(X:Y|A)
      +
      (2^{s-1} - 1)I(X:Y|B)\\
    &\quad+
      (s - 1)2^{s-2}\big(I(X:A|Y)+I(Y:A|X)\big)\\
    &\quad+
      ((s - 3)2^{s-2} + 1)\big(I(X:B|Y)+I(Y:B|X)\big)\\
    &\quad+
      (2^{s-1} - 1)I(A:B)\quad\text{for $s\in\Nbb$}
  \end{align*}
  by a suitable choice of the parameter $s$.
\end{remark}

In this paper we present more direct proofs of Theorems~\ref{th:sqrt},~\ref{th:sqrt-ii}, and~\ref{th:eloge}  without
deducing preliminary infinite families of non-Shannon--type inequalities for entropy.

\smallskip

Another result proved in this paper reveals more subtle properties of probability distributions.
It can be presented as the following alternative. If in a six-tuple of random variables $(X,Y,A,B,C,D)$ the four quantities 
\[
 I(X:A|Y), \ I(Y:A|X), \  I(X:B|Y), \ I(Y:B|X)
\]
are negligibly small, then one of the two properties must be true: 
a stronger version of the Ingleton inequality for $(X,Y,A,B)$, 
\[
  I(X:Y) \textrm{ is much less than } I(X:Y|A) + I(X:Y|B) + I(A:B),   
\]
or a slightly weaker version of the Ingleton inequality for
$(X,Y,C,D)$, 
\[
  I(X:Y)\   \textrm{ is not much bigger than } I(X:Y|C) + I(X:Y|D) + I(C:D) .
\]
Here is the exact statement.  
\begin{theorem}[name=,restate=thsubtler]
  \label{th:subtler}
If the joint distribution $(X,Y,A,B,C,D)$ satisfies
\[
  \max\{I(X:A|Y), I(Y:A|X), I(X:B|Y), I(Y:B|X)\}
  \leq \epsilon\cdot I(X:Y)
\]
Then
\begin{enumerate}[label=(\alph*)]
\item\label{th:subtler-a}
  The following alternative holds: for every threshold
  $0< t < 1$  we have a \emph{stronger} version of Ingleton's
  inequality for $(A,B,X,Y)$
  \[
    I(X:Y) \le I(X:Y|A) + I(X:Y|B) + I(A:B)+
    \big(O(\sqrt{\epsilon})-t\big)I(X:Y) 
  \]
  or a \emph{weaker} version of Ingleton's inequality for $(X,Y,C,D)$
  \[
    I(X:Y) \le I(X:Y|C) + I(X:Y|D) + I(C:D) +
    \big(O(\sqrt{\epsilon})+t\big)I(X:Y)
  \]
\item\label{th:subtler-b}
  Besides, we have the following inequality for six random variables
  \begin{align*}
    \big(2-O(\sqrt{\epsilon})\big) \cdot I(X:Y)
    &\le
      I(X:Y|A) + I(X:Y|B) + I(A:B)\\
    &\;\;+
      I(X:Y|C) + I(X:Y|D) + I(C:D) 
  \end{align*}
\end{enumerate}
\end{theorem}
Note that the assumptions in the theorem above do not involve
variables $C,D$ in any capacity, so that the conclusion of the
theorem remains true for arbitrary $C$, $D$.%

\smallskip

While Theorems~\ref{th:sqrt},~\ref{th:sqrt-ii}, and~\ref{th:eloge} admit proofs based on explicitly derived infinite families of non-Shannon-type inequalities, 
for Theorem~\ref{th:subtler} (proved in Section~\ref{s:subtler}) 
we are are not aware of any alternative proof that differs substantially from our approach.

\smallskip

In the next section we give a brief overview of our technique.

\subsection{Methods and ideas}

This article presents more direct proofs of
Theorems~\ref{th:sqrt},~\ref{th:sqrt-ii} (Section~\ref{s:sqrt}), and Theorem~\ref{th:eloge}  (Section~\ref{s:eloge}) without
deducing preliminary infinite families of non-Shannon--type
inequalities for entropy.  
This approach offers, in our view, several
noteworthy advantages.  To begin with, it is conceptually more
intuitive, as it is grounded in certain structural properties of
probability distributions and in the intuition of \emph{materializable
  common information}, which we elaborate on below.  This perspective
enables us to reveal and prove properties that are stronger than a
mere conditional form of the Ingleton inequality, as in 
Theorem~\ref{th:subtler}.  In addition, the structural property
we identify appears particularly amenable to computer-assisted proof
systems.  We believe that this may enable computer-assisted  
discoveries of new conditional information inequalities, in much the
same spirit as non-Shannon--type inequalities and their applications
have been established in~\cite{dougherty2011non} and especially in
\cite{farras2018improving,gurpinar2019use}.

We now outline the \emph{structural property} of a distribution that is
at the core of all our proofs.  Somewhat informally, it states that if
$I(X:A|Y)$ and $I(Y:A|X)$ are sufficiently small, then a part of the
mutual information between $X$ and $Y$ can be ``materialized'' in a
sense analogous to the notion of Gács--K\"orner common information
\cite{gacs1973common}.  This means that there exists a random variable
$W$ that captures a portion of the mutual information between $X$ and
$Y$. On the one hand, this $W$ has negligible conditional entropy
given either $X$ or $Y$; on the other hand, $W$ retains a significant
amount of information---conditioned on $W$, the mutual information
between $(X,Y)$ and $A$ becomes very small. 
This idea is formalized in Lemma~\ref{l:Wk-bounds} below, which is our main technical tool. 

Before Lemma~\ref{l:Wk-bounds}  is formulated, we need to introduce  two notions:  \emph{symmetry} among random variables 
 and \emph{tropical probability spaces}.
We begin with  the symmetry.
Suppose that a tuple
  $(X,Y,A,B,Z)$ is such  that the variables $A$ and $B$ take value in
  the same alphabet and such that the joint distribution $p$ of the
  tuple satisfies
  \[
    p(x,y,a,b,z)=p(x,y,b,a,z)
  \]
  for all samples $x,y,z$ in the alphabets of $X,Y,Z$, respectively,   and all $a,b$ in the (shared) alphabet of $A$ and $B$. 
  In that case we say that the tuple is \emph{symmetric} with respect to transposing
  $A$ and $B$ or variables $A$ and $B$ are symmetric within the
  tuple.

Tropical probability spaces were introduced in~\cite{matveev2018asymptotic}. 
\emph{Tropicalization} associates to a random variable $X$ (or a tuple of random variables) a sequence
  $X^{tr}:=(X^{(n)})_{n\in\Nbb_{0}}$, whose elements behave, in an appropriate sense, like the joint distribution composed of $n$ independent copies of $X$.
  The entropy
  of such a  \emph{tropical} random variable is defined as
  \[
    H(X^{tr}):=\lim_{n\to\infty}\tfrac1n H\big(X^{(n)}\big)
  \]
  with the obvious generalization for tuples and entropic expressions on tuples.
   This is similar to the \emph{tensorization} trick, which is widely
   used in analysis, geometry, probability, and related areas.  
   In Section~\ref{s:prelim-trop} we discuss the technique of tropicalization in more detail.
 
  One of the advantages of tropicalization is that any relation between entropic quantities 
  that holds for classical random variables up to a sublinear error holds exactly in the tropical setting. 
  From now on we omit the superscript $tr$ and simply refer to such objects as \emph{tropical}. 
  Properties of tuples are understood element-wise along the tropical sequence.
  For example, a tropical pair $(X,Y)$ is called symmetric if the pair $\big(X^{(n)},Y^{(n)}\big)$ is symmetric for all sufficiently large $n$. 
\begin{lemmalpha}[name=,restate=lWkbounds]
  \label{l:Wk-bounds}
  Suppose a tropical tuple $(X,Y,A,B,\dots)$ 
  is such that for some $\epsilon>0$
  \[
    \max\set{I(X:A|Y),I(Y:A|X)}\leq\epsilon\cdot I(X:Y)
  \]
  Then there exists a sequence of tropical random variables 
  $W_{i}$, $i\in
  \Nbb_{0}$ such that
  \begin{enumerate}[label=(\alph*)]
  \item  \label{l:Wk-bounds-wxy}
    $H(W_{i}|X,Y)=0$
  \item \label{l:Wk-bounds-I}
    $I(X:Y|W_{i+1})  = I(X:Y|W_{i},A) $ 
  \item \label{l:Wk-bounds-wx}
    $\max\set{H(W_{i}|X),H(W_{i}|Y)}\leq i\cdot\epsilon\cdot I(X:Y)$
  \end{enumerate}
  Moreover, if the variables $A$ and $B$ were symmetric in the
  original distribution $(X,Y,A,B,\ldots)$, then we may assume that
  this property is preserved for the extended distribution.
\end{lemmalpha}
\begin{corollaryalpha}[name=,restate=corWbound]
  \label{cor:W-bound}
  Suppose a tropical tuple $(X,Y,A,B,\dots)$
  is such that for some $\epsilon>0$
  \[
    \max\set{I(X:A|Y),I(Y:A|X)}\leq\epsilon\cdot I(X:Y)
  \]
  Then there exists a tropical $W$ such that
   \begin{enumerate}[label=(\alph*)]
  \item \label{cor:W-bound-a}
    $H(W | X,Y)=0$
  \item \label{cor:W-bound-b}
    $I(X,Y:A| W )  = O(\sqrt{\epsilon}) \cdot I(X:Y)$ 
  \item \label{cor:W-bound-c}
    $\max\set{H(W | X),H(W | Y)}= O(\sqrt{\epsilon}) \cdot I(X:Y)$
  \end{enumerate}
\end{corollaryalpha}
This implies immediately another corollary:
\begin{corollaryalpha}
  \label{cor:W-bound-0}
Lemma~\ref{l:studeny} holds for almost-entropic points.
\end{corollaryalpha}

The tropicalization procedure provides a
convenient “interface’’ that hides certain technical steps --- such as
scaling up the distribution and controlling sublinear slack terms --- and
thereby allows us to give a cleaner and more usable formulation of the
technical Lemma~\ref{l:Wk-bounds} and
Corollary~\ref{cor:W-bound}.  
By contrast, our main results
(Theorems~\ref{th:sqrt},~\ref{th:sqrt-ii},~\ref{th:eloge},
and~\ref{th:subtler}), which rely on Lemma~\ref{l:Wk-bounds} and
Corollary~\ref{cor:W-bound}, can be understood in the more usual way:
they apply directly to conventional tuples of random variables.

We prove  Lemma~\ref{l:Wk-bounds} and
Corollary~\ref{cor:W-bound} in Section~\ref{s:construction}.
The proofs  may be viewed as compactly capturing the
underlying construction originally due to Mat\'u\v s.  In particular,
the proof of Corollary~\ref{cor:W-bound} requires an \emph{a priori}
unbounded number of applications of the Copy Lemma.  This fact
effectively encodes an infinite family of non-Shannon--type information
inequalities, even though they never appear explicitly in the
argument.

At a high level, the proofs of Theorems~\ref{th:sqrt}
and~\ref{th:sqrt-ii} follow a similar pattern.  We first apply
Corollary~\ref{cor:W-bound}, and then use standard Shannon--type
inequalities on the original variables together with the auxiliary
variable $W$; see the full argument in Section~\ref{s:sqrt}.  
The proofs presented in this paper are entirely conventional in the
sense that we explicitly write down the necessary entropy
inequalities.  However, the same arguments can be generated or
mechanically verified with computer assistance.  In practice, one may
apply Corollary~\ref{cor:W-bound} in a suitable form, take the
properties (a)--(c) for the resulting auxiliary variable~$W$, include
all basic Shannon-type inequalities involving the variables under
consideration (including $W$), and then use a linear-programming
solver to deduce the desired conclusion.  Since the bounds in
Corollary~\ref{cor:W-bound} hold only up to $O(\sqrt{\epsilon})$, the
resulting inequality must likewise be understood up to an error term
of order~$O(\sqrt{\epsilon})$.

In the proof of  Theorem~\ref{th:eloge}  (see Section~\ref{s:eloge}) 
and Theorem~\ref{th:subtler} (see Section~\ref{s:subtler})
we combine Lemma~\ref{l:Wk-bounds} with symmetry considerations.

\section{Preliminaries}\label{s:prelim}
\subsection{Notation}
In what follows we use the following notation and conventions. We
denote $[n]:=\set{0,1,\dots,n-1}$ and $2^{[n]}$ will stand for the
power set of $\set{0,1,\dots,n-1}$.
In this paper, when we talk about \emph{random variables},  we always mean finitely supported random variable.
 For a collection of jointly
distributed random variables $(X_{0},\dots,X_{n-1})$ and
$J\in2^{[n]}$, we denote by $X_{J}$ the joint random variable
corresponding to the subcollection $(X_{i}\st i\in J)$. The entropy of
a random variable $X$ will be denoted $H(X)$.  
We use the standard notation for the conditioned
  entropy, the mutual information in pairs and triples of random
variables,
\begin{align*}
  &H(X|Y) := H(XY) - H(Y)\\
  &I(X:Y) := H(X)-H(X|Y),\\
  &I(X:Y|Z) := H(X|Z)-H(X|Y,Z),\\
  &I(X:Y:Z) := I(X:Y) - I(X:Y|Z),\\
  &I(X:Y:Z|W) := I(X:Y|W) - I(X:Y|W,Z).
\end{align*}

\subsection{Tropicalization/Tensorization}
\label{s:prelim-trop}
In this article, we systematically use the concept of
\emph{tropicalization}/\emph{tensorization} of random variables.  This
is a useful tool to investigate the asymptotic behavior of entropy
profiles and to prove entropic inequalities for individual tuples of
random variables.

The idea of tropicalization is the following. Given an $n$-tuple
$(X_{i}:i=0,\dots,n-1)$ of random variables, we replace it with the
sequence of its \textit{tensor powers}
$(X_{i}^{\otimes k}:i=0,\dots,n-1)_{k=1}^{\infty}$, where
$X^{\otimes k}$ stands for the joint of $k$ independent copies of $X$.
For such sequences we can define all the usual notions and operations,
such as entropy profile, joints, etc. by performing these operations
elementwise in the sequence and normalizing.

However, the class of tensor-power sequences as above is too restrictive,
and we relax the definition in the following way: we consider sequences
of $n$-tuples
\[
  (\Xbf^{(k)})_{k=1}^{\infty}=(X_{i}^{(k)}\st
  i=0,\dots,n-1)_{k=1}^{\infty}
\]
that do not deviate from tensor-powers-sequences too much, more
specifically
\[
  \dist(\Xbf^{(k)}\otimes\Xbf^{(l)},\Xbf^{(k+l)})\leq C\cdot\psi(k+l)
\]
where $\dist$ stands for the entropic distance between $n$-tuples of
random variables and $\psi$ is an \emph{admissible}%
\footnote{A function $\psi$ is \emph{admissible} if it is
    monotone and sufficiently slowly growing, more specifically
    $\int_{1}^{\infty}(\psi(t)/t^{2})\mathrm{d} t<\infty$ } %
function, and $C$ is some
positive constant that may depend on the sequence
  $(\Xbf^{(k)})_{k\in\Nbb}$. For example, we can take
$\psi(t):=t^{3/4}$ and fix this choice in this article;
see~\cite{matveev2020tropical} for details.

We say that two such sequences $\Xbf^{(k)}=(X_{i}^{(k)}\st i=0,\dots,n-1)$ and
$\Ybf^{(k)}=(Y_{i}^{(k)}\st i=0,\dots,n-1)$ are equivalent if
\[
  \lim_{k\to\infty}\frac1k\dist(\Xbf^{(k)},\Ybf^{(k)})=0
\]
We call equivalence classes of
such sequences \textit{tropical $n$-tuples of random variables}.

From now on, when we say a (tuple of) random variables we mean tropical
(tuple of) random variables.  Entropy of such a tropical random
variable is defined as the limit of scaled entropies of individual
members of the sequence.
Other properties of tropical tuples, such as symmetry, are to be
  understood in the sense that they apply to the tropical sequence
  element-wise for sufficiently large $n$

\begin{remark}\label{r:trop2class2trop}
  Two observations are in order.
  \begin{enumerate}[label=(\roman*)]
  \item First of all, any linear inequality that is valid for
    entropies for all tropical random variables implies the same
    inequality for entropies of conventional random variables.
  \item Secondly, if every component of a tropical tuple obeys a given
    linear inequality with a sublinear error term, e.g., up to a constant independent of $k$, 
     then the tropical tuple itself satisfies this
    inequality exactly (with no error).
  \end{enumerate}
  Thus, these ``tropical''
  objects can be interpreted as a helpful tool to study properties of
  the more conventional objects such as probability distributions.
\end{remark}

\subsection{Symmetries in distribution and tropical random variables} 
\label{ss:symmetries}
In this section we show how to transform any tropical tuple
$(X_{0},\dots,X_{n-1},A,B)$ into another tuple
$(X_{0},\dots,X_{n-1},\Abar,\Bbar)$ that is symmetric with respect to
the transposition of $\Abar$ and $\Bbar$.  This observation will allow us
in some cases to assume, without loss of generality, that the
distributions involved in the argument satisfy the condition of
symmetry.
\begin{lemma}
  \label{l:symmetrization}
  For every tropical tuple $(X_{0},\dots,X_{n-1},A,B)$ there
  exists another tropical tuple $(X_{0},\dots,X_{n-1},\Abar,\Bbar)$
  such that
  \begin{itemize}
  \item the probability distributions in the new tuple are symmetric
    with respect to the transposition of $\Abar$ and $\Bbar$, and
  \item for any $J\in2^{[n]}$ one has
    \begin{align*}
      &H(X_{J},\Abar)
      =
        H(X_{J},\Bbar)=\frac12\big(H(X_{J},A)+H(X_{J},B)\big)
      &\text{and}\\ 
      &H(X_{J},\Abar,\Bbar)
      = H(X_{J},A,B)
    \end{align*}
  \end{itemize}
  Thus, any linear combination of entropies of
  $(X_{0},\dots,X_{n-1},A,B)$ that is symmetric with respect to
  swapping $A$ and $B$ will remain unchanged when $(A,B)$ is replaced
  by $(\Abar,\Bbar)$.
  \end{lemma}
\begin{proof}
  Let $(X_{0},\dots,X_{n-1},A,B)$ be an $(n+2)$-tuple of random
  variables (not tropicalized).  Denote by $\Acal$ and $\Bcal$
  the alphabets of $A$ and $B$, respectively, and by $p$ the joint
  distribution of the whole $(n+2)$-tuple.  We tacitly assume that
  $\Acal$ and $\Bcal$ are disjoint. Define two new random variables
  $\Abar$ and $\Bbar$ jointly distributed with $X_{i}$'s, each taking
  values in $\Acal\cup\Bcal$ and distributed according to the law
  $\pbar$ defined by
  \[
    \pbar(x_{0},\dots,x_{n-1},a,b)=
    \begin{cases}
      \frac12 p(x_{0},\dots,x_{n-1},a,b)
      &
        \text{if $(a,b)\in\Acal\times\Bcal$}
      \\
      \frac12 p(x_{0},\dots,x_{n-1},b,a)
      &
        \text{if $(b,a)\in\Acal\times\Bcal$}
      \\
      0&\text{otherwise}
    \end{cases}
  \]
  Distribution $\pbar$ is symmetric with respect to swapping $\Abar$
  and $\Bbar$ and for any $J \in 2^{[n]}$ holds
  \[
    H(X_{J},\Abar)=H(X_{J},\Bbar)=\frac12(H(X_{J},A)+H(X_{J},B))+\log2
  \]
  and
  \[
    H(X_{J},\Abar,\Bbar)=H(X_{J},A,B)+\log2
  \]
  Thus, we symmetrize the distribution, while changing any $(A,B)$-symmetric
  entropic expressions by $O(1)$, uniformly with respect to the data.
  Given a tropical tuple we can perform the above procedure elementwise
  in the sequence of tuples, thereby obtaining exact equalities for the
  symmetric entropic expressions, see Remark~\ref{r:trop2class2trop}.
\end{proof}

\subsection{Copy lemma}\label{ss:copy-lemma}
In this section, we recall the classical \emph{Copy Lemma}.
Essentially, it states that for every distribution $(X, Y, A)$, one
can construct an extended distribution $(X, Y, A,A')$ such that the
conditional distribution of $A'$ given $X$ coincides with that of
$A$ given $X$, and moreover, the random variables form the Markov
chain
\[
  (Y, A) \to X \to A'.
\]
In other words, $A'$ is an ``independent copy'' of $A$ that preserves
the same joint distribution with $X$.  The lemma extends to the
setting where $X$ and $Y$ are not individual random variables but
jointly distributed tuples.

Despite its simplicity, this lemma remains a powerful tool in proving
non-Shannon--type information inequalities for entropy. This
construction appeared implicitly in the original work of Zhang and
Yeung~\cite{zhang1997non,zhang1998characterization}, and was later
formulated explicitly in~\cite{dougherty2006six}.

\begin{lemma}[Copy Lemma]
  Every tuple of jointly distributed  random variables
  \[
    (X_{i},Y_{j},A\st i=0,\ldots,k-1;\;j=0,\dots,\ell-1)
  \] 
  can be extended to a longer tuple
  \[
    (X_{i},Y_{j},A,A'\st i=0,\ldots,k-1;\;j=0,\dots,\ell-1)
  \] 
  such that $(X_{0},\ldots,X_{k-1},A)$ and $(X_{0},\ldots, X_{k-1},A')$
  have the same distribution, and $A'$ is independent of
  $(Y_0,\ldots,Y_{\ell-1},A)$ conditioned on $(X_0,\ldots, X_{k-1})$.
  
  Moreover, if the original distribution of
  $(X_{[k]},Y_{[\ell]},A) $ was symmetric with respect to some
  permutation of variables $(Y_{j},A)$ and/or permutation of variables
  $(X_{i})$, then the extended distribution of
  $(X_{[k]},Y_{[\ell]},A,A')$ will possess the same symmetries.
\end{lemma}
\begin{proof}
  The new variable $A'$ will have the same alphabet as $A$, and the new
  probability distribution can be defined by the following
  equation: for all instances $a$ and $a'$ of $A$ and $A'$,
  respectively,
  \begin{align*}
    \pr[(X_{[k]},&Y_{[\ell]},A,A') =  (\xbf,\ybf,a,a')]=\\
    &=
    \begin{cases}
      \frac{
      \pr[(X_{[k]},Y_{[\ell]},A) = (\xbf,\ybf,a)] \cdot
      \pr[(X_{[k]},A) = (\xbf,a')]
      }
      {\pr[X_{[k]} = \xbf]}
      & \text{if $\pr[X_{[k]} =  \xbf]>0$,}
      \\
      0 & \text{otherwise}.
    \end{cases}
  \end{align*}
  where $\xbf=(x_{0},\dots,x_{k-1})$, $\ybf=(y_{0},\dots,y_{\ell-1})$ are
  instances of $X_{[k]}$ and $Y_{[\ell]}$, respectively.
  
  The \emph{moreover} claim of the lemma follows directly from
  the definition of the new distribution.
\end{proof}
The entropy profile of the extended tuple
$(X_{[k]},Y_{[\ell]},A, A')$ is not totally determined by the
entropies of the original distribution $(X_{[k]},Y_{[\ell]},A)$.
However, some components of the new entropy profile are determined by
the definition of $A'$: for all $I\subset [k]$ and
$J\subset [\ell]$
\begin{align*}
  H(X_{I},A') &= H(X_{I},A),\\
  H(Y_{J},A,A'\mid X_{[k]}) &= H(Y_{J},A\mid X_{[k]})+H(A' \mid X_{[k]}).
\end{align*}

\subsection{Ahlswede--K\"orner reduction}
The Ahlswede--K\"orner reduction modifies a given tropical $n$-tuple so
that some entropies in the new tuple vanish, while some (combinations
of) other entropies stay the same.
\begin{lemma}[Ahlswede--K\"orner reduction]
\label{l:ak}
Let us consider a tropical $n$-tuple $(X_{i}:i=0,\dots,n-1)$ and a subcollection
$(X_{i}:i=0,\dots,k-1)$ of it, where $k< n$. Then there exists an
$(n+1)$-tuple $(X_{0},\dots,X_{n-1},X'_{n-1})$, such that
\begin{equation*}
  \begin{cases}
    H(X'_{n-1}|X_{0},\dots,X_{k-1})=0\\
    H(X_{J}|X'_{n-1})=H(X_{J}|X_{n-1})&\text{for any $J\in2^{[k]}$}
  \end{cases}
\end{equation*}
We will call the new random variable $X'_{n-1}$ the Ahlswede--K\"orner
reduction of $X_{n-1}$ under $(X_{0},\dots,X_{k-1})$ and denote the new variable 
\[
  X'_{n-1}:=\AK(X_{n-1}\mid X_{0},\dots,X_{k-1}).
\]
\end{lemma}
\begin{remark}
\label{rem:AK-reformulated}
The claim $ H(X_{J}|X'_{n-1})=H(X_{J}|X_{n-1})$ in the Ahlswede–K\"orner reduction
can be reformulated as $ I(X_{J}:X'_{n-1})=I(X_{J}:X_{n-1})$.
\end{remark}
\begin{remark}
The extension of the distribution guaranteed by the Ahlswede--K\"orner reduction is not unique, i.e., there may exist different (nonisomorphic) extensions of $(X_{i}:i=0,\dots,n-1)$ 
satisfying the conditions of Lemma~\ref{l:ak}.
We abuse notation and write $X'_{n-1}:=\AK(X_{n-1}\mid X_{0},\dots,X_{k-1})$, meaning that any such extension may be chosen.
\end{remark}

Essentially Lemma~\ref{l:ak} claims that the variable $X_{n-1}$
  can be ``reduced'' to a new random variable $X'_{n-1}$, that will be
  functionally dependent on $X_{[k]}$, while all the quantities
  $I(X_{J}:X_{n-1})$, $J\in 2^{[k]}$, will remain preserved after the
  reduction. 
  In Fig.~\ref{fig:ak} we illustrate this reduction for the case $k=n=2$.
   This statement is not literally true for untropicalized
  tuple, but only hold up to an additive error, which is sublinear
  with respect to $H(X_{[n]})$. Recall that tropical tuple is a
  growing sequence of the classical tuples of random variables such
  that entropic quantities of individual members of the sequence grow
  quasi-linearly. Entropic quantity evaluated on the tropical
  tuple is the limit of the corresponding quantities of individual
  tuples normalized by the index in the sequence. Thus after
  tropicalization sublinear errors vanish.
  
  \def\firstcircle{(0,0) circle (1.8cm)}
\def\secondcircle{(0:2cm) circle (1.8cm)}
\def\thirdcircle{(-60:2cm) circle (1.8cm)}

\begin{figure}
\begin{center}
\begin{tikzpicture}[scale=1.0]
    \begin{scope}[fill opacity=0.5]
        \fill[red] \firstcircle;
        \fill[blue] \secondcircle;
        \fill[gray] \thirdcircle;
    \end{scope}

        \draw[thick] \firstcircle node[below] {};
        \draw[thick] \secondcircle node [below] {};
	\draw[thick] \thirdcircle node [below] {};

	\draw[thick] node at (-0.8,0.3) {\tiny $H(X_0|X_1,X_2)$};
	\draw[thick] node at (2.8,0.3) {\tiny $H(X_1|X_0,X_2)$};
	\draw[thick] node at (1,-2.6) {\tiny $H(X_2|X_0,X_1)$};
	\draw[thick] node at (1.07,-0.55) [rotate=30, scale=0.65] {\tiny $I(X_0:X_1:X_2)$};
	\draw[thick]  node at (0.88,0.65) [rotate=50,  scale=0.8]  {\tiny $I(X_0:X_1|X_2)$};
	\draw[thick]  node at (1.90,-1.20)  [rotate=40,  scale=0.7] {\tiny $I(X_1:X_2|X_0)$};
	\draw[thick]  node at (0.1,-1.20) [rotate=-40,  scale=0.7]  {\tiny $I(X_0:X_2|X_1)$};

	\draw[thick] node at (-1.5,1.5) {$X_0$};
	\draw[thick] node at (3.5,1.5) {$X_1$};
	\draw[thick] node at (1,-3.9) {$X_2$};

 \begin{scope}[shift={(10,0)}]

\begin{scope}[fill opacity=0.5]
    \fill[red] \firstcircle;
    \fill[blue] \secondcircle;
\end{scope}

\begin{scope}
    \clip \firstcircle;
    \fill[gray, fill opacity=0.5] \thirdcircle;
\end{scope}

\begin{scope}
    \clip \secondcircle;
    \fill[gray, fill opacity=0.5] \thirdcircle;
\end{scope}
        
        \draw[thick] \firstcircle node[below] {};
        \draw[thick] \secondcircle node [below] {};
	\draw[thick] \thirdcircle node [below] {};

	\draw[thick] node at (-0.8,0.3) {\tiny $H(X_0|X_1,X_2)$};
	\draw[thick] node at (2.8,0.3) {\tiny $H(X_1|X_0,X_2)$};
	\draw[thick] node at (1,-2.6) {\small $0$};
	\draw[thick] node at (1.07,-0.55) [rotate=30, scale=0.65] {\tiny $I(X_0:X_1:X_2)$};
	\draw[thick]  node at (0.88,0.65) [rotate=50,  scale=0.8]  {\tiny $I(X_0:X_1|X_2)$};
	\draw[thick]  node at (1.90,-1.20)  [rotate=40,  scale=0.7] {\tiny $I(X_1:X_2|X_0)$};
	\draw[thick]  node at (0.1,-1.20) [rotate=-40,  scale=0.7]  {\tiny $I(X_0:X_2|X_1)$};

	\draw[thick] node at (-1.5,1.5) {$X_0$};
	\draw[thick] node at (3.5,1.5) {$X_1$};
	\draw[thick] node at (1,-3.9) {$X_2'$};

\end{scope}

	\draw[thick]  (6.1,-2.0) node {\Large $\xRightarrow{\text{AK reduction}}$};

\end{tikzpicture}
\caption{Diagrams representing the Ahlswede--K\"orner reduction for $k=2$. 
The reduced variable $X_2' = \AK(X_{n-1}\mid X_{0},X_{1})$ has the same mutual information with $X_0$, $X_1$, and $(X_0,X_1)$ as the original variable $X_2$;
at the same time, $H(X_2'|X_0,X_1)$ vanishes.}\label{fig:ak}
\end{center}
\end{figure}

\paragraph{History of the Ahlswede--K\"orner reduction.}
This technical tool goes back to the work of Ahlswede and K\"orner in
the 1970s \cite{ahlswede2006common,Ahlswede-Connection-1977}.  For
readers interested in a proof of this lemma, we can point to several
sources.
For the case $k=2$ and $n=3$, the statement is proven in
\cite[Lemma~5]{makarychev2002new}, and the argument extends
straightforwardly to arbitrary $k$ and $n$.  A similar argument can be
found in \cite[Theorem~3.1]{matveev2019arrow}. The latter article,
however, uses somewhat different terminology.

The lemma can also be obtained as a special case of more general
techniques.  Namely, one may first introduce a copy $X'_{n-1}$ of
$X_{n-1}$ (preserving its joint distribution with
$(X_{0},\dots,X_{k-1})$), and then perform tightening as in
\cite[Lemma~3]{matus2016entropy} (with $L=\{X'_{n-1}\}$ and
$t = H(X_{n-1} \mid X_{0},\dots,X_{k-1})$).  The technically most
involved part of this argument is established in
\cite[Theorem~3]{matus2007two}, which can be viewed as a
generalization of the Slepian--Wolf theorem \cite{slepian1973coding}.
An overview of the tightening technique and its applications can be
found in \cite{csirmaz2025exploring}.

\section{Main construction}\label{s:construction}
In this section we prove our main technical results,
Lemma~\ref{l:Wk-bounds} and Corollary~\ref{cor:W-bound}, which are
repeated below verbatim. They are systematically applied in all proofs of
the article.  Assume we are given a jointly distributed triple $(X,Y,A)$
or, possibly, a longer joint distribution
$(X,Y,A,B,C_0,\ldots,C_{n-1})$ extending $(X,Y,A)$.  For each integer
$k\ge 0$, we extend this tuple with a new variable $W_{k}$ satisfying the
properties stated in the lemma.

\lWkbounds*

\subsection{The main construction}
We start with a tuple of random variables $(X,Y,A,B,C_{0},\dots,C_{n-1})$,
possibly satisfying some symmetry that permutes variables $A$, $B$ and
$C_{i}$'s, and then apply the following steps.
\paragraph{Step 1 (duplication):} %
First of all, we apply  the Copy Lemma $k$ times to $A$ under $(X,Y)$
and obtain new variables $A'_1,\ldots, A'_k$ such that for each  $i=1,\ldots, k$  
\begin{itemize}
\item the new variable $A'_i$ has the same distribution as $A$ conditional on $(X,Y)$;
\item conditional on $(X,Y)$, the new variable $A'_i$ is independent of  $A'_1,\ldots,A'_{i-1},A,B,C_{0},\dots,C_{n-1}$.
\end{itemize}
Observe that
\begin{align*}
I(A'_1,\ldots,A'_k  &: A,B,C_{0},\dots,C_{n-1} |X,Y) \\
&= 
\sum\limits_{i=1}^k I(A'_{i}  : A,B,C_{0},\dots,C_{n-1} |X,Y,A'_1,\ldots,A'_{i-1} ) \\
& \le \sum\limits_{i=1}^k I(A'_{i}  : A'_1,\ldots,A'_{i-1}, A,B,C_{0},\dots,C_{n-1} |X,Y ) 
=0
\end{align*}
(the first step is the chain rule for the mutual information; 
the second step is an instance of the Shannon-type inequality $I(U:V|W) \le I(U:V,W)$; the last equality follows from the definition of $A'_i$ in the Copy Lemma).
Thus, we conclude that 
\begin{itemize}
\item the random variables $A'_1,\ldots, A'_k$ and the tuple
  $(A,B,C_0,\ldots,C_{n-1})$ are all together mutually independent
  conditional on $(X,Y)$.
\end{itemize}
Denote $\tilde A_{k} = (A'_1,\ldots,A'_k)$.
Observe that if the original distribution was symmetric under swapping
$A$ and $B$, then the same symmetry holds for the extended
distribution
 \[
   (X,Y, A,B,C_0,\ldots,C_{n-1},\tilde A_{k})
 \]
 (see the \emph{moreover} claim of the Copy Lemma).

\paragraph{Step 2 (reduction):} %
Further, we apply the Ahlswede--K\"orner reduction to $\tilde A_{k}$
under the tuple $(A,B,C_0,\ldots,C_{n-1})$ and denote the result by
\begin{equation*}
  W_{k}:=\AK(\tilde A_{k}\mid X,Y,A,B,C_0,\ldots,C_{n-1}).
\end{equation*}
We use the convention that $W_{0}$ is the trivial (valued in a
one-point-set) random variable.

Observe that if $(X,Y, A,B,C_0,\ldots,C_{n-1},\tilde A)$ is symmetric
under swapping $A$ and $B$, then the \emph{entropy profile} of
$(X,Y, A,B,C_0,\ldots,C_{n-1},W_{k})$ has a similar symmetry, i.e.,
\[
  H(A,W_{k}) = H(B,W_{k}), \quad H(A, C_J,W_{k}) = H(B,C_J, W_{k}),
  \text{ etc}.
\]

\subsection{Proof of Lemma \ref{l:Wk-bounds}}
We recall that we define each variable $A'_{i+1}$ with help of the Copy Lemma so that
\begin{enumerate}[label=(\alph*)]
\item it has the same distribution as $A$  conditional on $(X,Y)$;
\item $A'_{i+1}$  and $(A'_1,\ldots, A'_{i},A,B,C_0,\ldots,C_{n-1})$ are independent conditional on $(X,Y)$.
\end{enumerate}
\begin{claim}
\label{cl:markov}
For each $k$ 
\[
(A,B,C_0,\ldots,C_{n-1})\to(X,Y) \to (A'_1,\ldots, A'_{k})
\]  
is a Markov chain. 
\end{claim}
This lemma follows easily from the definition of the quantities  $A'_i$. 
For completeness of exposition, we give a detailed derivation in terms of information inequalities.
\begin{proof}[Proof of the claim:]
The argument is based on the property (b) above with $i=1,\ldots, k-1$.
We use the chain rule, and then apply $i$  times  inequalities having the form \eqref{eq:chain-rule-simplified} and obtain 
\begin{align*}
I(A,B,C_0,\ldots,C_{n-1}:A'_1,\ldots, A'_i |X,Y) & = \sum\limits_{j=0}^i I(A,B,C_0,\ldots,C_{n-1}:A'_j |X,Y,A'_1,\ldots,A'_{j-1}) \\
&\le \sum\limits_{j=0}^i I(A'_1,\ldots,A'_{j-1},A,B,C_0,\ldots,C_{n-1}:A'_j |X,Y) \\
&= 0,
\end{align*}
and the claim is proven.
\end{proof}

Now we can prove assertion~\ref{l:Wk-bounds-wxy} of Lemma~\ref{l:Wk-bounds}.
Claim~\ref{cl:markov} implies
\begin{align*}
I(\Atilde_{k}:A,B,C_{0},\ldots, C_{n-1},X,Y) &\le I(\Atilde_{k}:X,Y) +  I(\Atilde_{k}:A,B,C_{0},\ldots, C_{n-1}|X,Y)  \\
&=  I(\Atilde_{k}:X,Y).
\end{align*}
From the definition of the Ahlswede--K\"orner reduction (see Remark~\ref{rem:AK-reformulated}) it follows 
\begin{align*}
I(W_{k}:X,Y,A,B,C_{0},\ldots, C_{n-1})  &=
I(\Atilde_{k}:X,Y,A,B,C_{0},\ldots, C_{n-1}) \\ 
&\le  I(\Atilde_{k}:X,Y)
=  I(W_k:X,Y).
\end{align*}
Therefore,  with the chain rule we obtain 
\begin{align*}
 I(W_{k}:A,B,C_{0},\ldots, C_{n-1}|X,Y) &= I(W_{k}:X,Y,A,B,C_{0},\ldots, C_{n-1})  - I(W_k:X,Y) \\
 &=0.
\end{align*}
Hence,
\begin{align*}
H(W_k|X,Y) &= H(W_{k}|X,Y,A,B,C_{0},\ldots, C_{n-1}) + I(W_{k}:A,B,C_{0},\ldots, C_{n-1}|X,Y) \\
& =  H(W_{k}|X,Y,A,B,C_{0},\ldots, C_{n-1}) = 0
\end{align*}
(the last equality follows from the Ahlswede--K\"orner reduction).

\smallskip

To prove assertion~\ref{l:Wk-bounds-I} of Lemma~\ref{l:Wk-bounds}, we observe that the following conditional distributions are all identical:
\begin{enumerate}[label=(\roman*)]
\item distribution of $A$  given the value of $(X,Y,A'_1,\ldots, A'_{i})$ 
\item distribution of $A$  given the value of $(X,Y)$  

(this distribution is equivalent to the previous one due to the Markov property)
\item distribution of $A'_{i+1}$  given the value of $(X,Y)$ 

(equivalent to the previous one due to the property (a) above)
\item distribution of $A'_{i+1}$  given the value of $(X,Y,A'_1,\ldots, A'_{i})$ 

(equivalent to the previous one due to the property (b) above).
\end{enumerate}
Comparing (i) and (iv),  we conclude that the conditional distributions
\[
(X,Y) \ \text{given}\ \Atilde_{i+1} 
\]
and 
\[
 (X,Y)\ \text{given}\ (\Atilde_i,A)
\]
are the same. It follows that 
\[
\left\{
\begin{array}{l}
H(X,Y|\Atilde_{i+1}) = H(X,Y|\Atilde_i,A) = H(X,Y,A|\Atilde_i) - H(A|\Atilde_i),\\
H(X|\Atilde_{i+1}) = H(X|\Atilde_i,A)  = H(X,A|\Atilde_i) - H(A|\Atilde_i),\\
H(Y|\Atilde_{i+1}) = H(Y|\Atilde_i,A)  = H(Y,A|\Atilde_i) - H(A|\Atilde_i).
\end{array}
\right.
\]

The Ahlswede--K\"orner reduction guarantees that 
\[
\left\{
\begin{array}{l}
H(X,Y|\Atilde_{i+1}) = H(X,Y|W_{i+1}),\\
H(X|\Atilde_{i+1}) = H(X|W_{i+1}),\\
H(Y|\Atilde_{i+1}) = H(Y|W_{i+1}),\\
\end{array}
\right.
\] 
and
\[
\left\{
\begin{array}{l}
H(X,Y|\Atilde_{i}) = H(X,Y|W_{i}),\\
H(X|\Atilde_{i}) = H(X|W_{i}),\\
H(Y|\Atilde_{i}) = H(Y|W_{i}), \\
H(X,Y,A|\Atilde_i) = H(X,Y,A|W_i), \\
H(X,A|\Atilde_i) = H(X,A|W_i), \\
H(Y,A|\Atilde_i) = H(Y,A|W_i), \\
H(A|\Atilde_i) = H(A|W_i). 
\end{array}
\right.
\] 

We combine these observations and obtain
\begin{align*}
I(X:Y|W_{i+1}) & = H(X|W_{i+1}) + H(Y|W_{i+1}) - H(X,Y|W_{i+1})\\
&= H(X|\Atilde_{i+1}) + H(Y|\Atilde_{i+1}) - H(X,Y|\Atilde_{i+1})\\
&= H(X|\Atilde_{i},A) + H(Y|\Atilde_{i},A) - H(X,Y|\Atilde_{i},A)\\
&=
H(X,A|\Atilde_i) +
H(Y,A|\Atilde_i)  -
H(X,Y,A|\Atilde_i) - H(A|\Atilde_i)\\
&=
H(X,A|W_i) +
H(Y,A|W_i)  -
H(X,Y,A|W_i) - H(A|W_i)\\
&= I(X:Y|W_i,A),
\end{align*}
and~\ref{l:Wk-bounds-I} is proven.

\smallskip

To prove~\ref{l:Wk-bounds-wx} of Lemma~\ref{l:Wk-bounds}, we observe that  
\begin{align*}
  H(W_{k} | X)
  &=
   H(W_{k},X) - H(X)\\
   &= H(X|W_{k}) + H(W_{k}) -H(X)  \\
   &= H(X|\Atilde_{k}) + I(\tilde A_{k} :X,Y,A,B,C_0,\ldots, C_{n-1} ) -H(X)  \\  
    &\text{/* from the AK Lemma */}\\
       &=  I(\tilde A_{k} :X,Y,A,B,C_0,\ldots, C_{n-1} ) - I(\Atilde_{k}:X)  \\  
         &=  I(\tilde A_{k} :X,Y ) - I(\Atilde_{k}:X)  \\  
             &\text{/* from the Copy Lemma:  $\Atilde_{k}$ is independent of}\\
  &\ \text{ $(A,B,C_{0},\ldots, C_{n-1})$ conditioned on $(X,Y)$  */}\\
         &=  I(\tilde A_{k} :Y |X) =  H(\tilde A_{k}  |X) - H(\tilde A_{k} |X,Y).
\end{align*}

Let us estimate the terms in the right-hand side of this equality.
On the one hand,
\[
  H(\tilde A_{k}  |X)   = H(A'_1,\ldots, A'_k  |X)
  \le \sum\limits_{i=1}^k H(A'_i|X) = k\cdot H(A|X). 
\]
On the other hand, by the definition of $A'_i$, we  have 
\[
  H(\tilde A_{k} |X,Y) = H(A'_1,\ldots, A'_k  |X,Y)
  = \sum\limits_{i=1}^k H(A'_i|X,Y) = k\cdot H(A|X,Y). 
\]
Thus,
\(
H(W_{k} | X)  = H(\tilde A_{k}  |X) - H(\tilde A_{k} |X,Y)
\)
rewrites to
\[
  H(W_{k} | X)  \leq k \cdot \left(H(A|X) - H(A|X,Y)  \right)
  = k \cdot I(Y:A|X) \le k\cdot\epsilon \cdot I(X:Y).
\]
A similar argument gives the bound for $ H(W_{k} | Y)$.

It remains to prove the \emph{moreover} part of the lemma.  We
already know that the \emph{entropy profile}
$(X,Y, A,B,C_0,\ldots,C_{n-1},W_{k})$ is symmetric under swapping $A$
and $B$.  
Applying
  Lemma~\ref{l:symmetrization}
  to the distribution proves
  the claim.

\subsection{Proof of Corollary~\ref{cor:W-bound}}
Here we deduce Corollary~\ref{cor:W-bound} from
Lemma~\ref{l:Wk-bounds}.
\corWbound*
\begin{proof}
  The argument is based on Lemma~\ref{l:Wk-bounds}, where we will
  choose a suitable $k$ and let $W=W_{k}$.  We apply
  Lemma~\ref{l:Wk-bounds} with the minimal $k$ such that
  \begin{equation*}
    I(X:Y|W_{k}) \le  I(X:Y|A,W_{k})  + 
    \delta\cdot I(X:Y)
  \end{equation*}
  (the choice of the parameter $\delta<1$ is explained later).
  This means that for each $i=1,2,\ldots, k$ the gap  
  \[
    I(X:Y| W_{i-1})-I(X:Y| W_{i})  
  \] 
  is at least $\delta\cdot I(X:Y)$.  Since the sum of the  gaps 
  \begin{align*}
  \big( I(X:Y)-I(X:Y| W_{1}) \big)   &+   \big( I(X:Y|W_1)-I(X:Y| W_{2}) \big)   + \ldots \\
  & \ldots +  \big( I(X:Y|W_{k-1})-I(X:Y| W_{k}) \big)   
   =  I(X:Y) -  I(X:Y|W_{k})
  \end{align*}
   does not exceed $I(X:Y)$, the number of steps $k$ is at most
  \(
  k\le 1/\delta.
  \)
  From Lemma~\ref{l:Wk-bounds}\ref{l:Wk-bounds-wx} we get
  \begin{equation*}
    \max\set{H(W_{k}|X),H(W_{k}|Y)}
    \leq
    k\cdot\epsilon\cdot I(X:Y)
    \leq
    \frac\epsilon\delta\cdot I(X:Y) 
  \end{equation*}
  It follows that 
  \begin{align*}
  I(&X,Y:A|W_{k})
  =
    I(X:A|Y,W_{k}) + I(Y:A|X,W_{k}) + I(X:Y:A|W_{k})\\
  & =
    I(X:A|Y,W_{k}) + I(Y:A|X,W_{k}) + I(X:Y:A|W_{k})\\
  &\le
    I(X:A|Y) +  H(W_{k}|Y) + I(Y:A|X)
  +
        H(W_{k}|X) + I(X:Y:A|W_{k}) \\
        &\ \text{(here we used twice \eqref{eq:shannon-type-strange-1-cond} with a suitable substitution of variables)}
    \\
  &\le
    I(X:A|Y) + I(Y:A|X) + I(X:Y:A|W_{k})
  +
    O\left(  \frac\epsilon\delta\right) \cdot I(X:Y)\\
  &\le
    I(X:Y:A|W_{k}) + O\left(  \frac\epsilon\delta\right) \cdot I(X:Y)\\ 
  &\le
    I(X:Y|W_{k}) - I(X:Y|A,W_{k})  +O\left(  \frac\epsilon\delta\right) \cdot I(X:Y)\\ 
  &\le
    \delta \cdot I(X:Y) + O\left(  \frac\epsilon\delta\right)\cdot I(X:Y) 
\end{align*}
We set $\delta=\sqrt{\epsilon}$, and the corollary follows.
\end{proof}

\section{Conditional inequalities with  precision $O(\sqrt{\epsilon})$}\label{s:sqrt}
In this section we prove that if the conditional information quantities  
\[
I(X:A|Y),\ I(Y:A|X)
\]
are of order $\epsilon$, then a version of the Ingleton inequality~\eqref{eq:ingleton} 
is valid with a  small \emph{correction term} of order
$\sqrt{\epsilon}$; see Theorem~\ref{th:sqrt}. 

\thsqrt*
\begin{proof}
  First of all, we apply Corollary~\ref{cor:W-bound} and take a $W$ such that 
  \begin{itemize}
  \item $H(W|X) = O(\sqrt{\epsilon})\cdot I(X:Y)$,
  \item $H(W|Y) = O(\sqrt{\epsilon})\cdot I(X:Y)$,
  \item $I(X,Y:A|W) = O(\sqrt{\epsilon})\cdot I(X:Y)$.
  \end{itemize}
  The rest of the proof is an application of suitable Shannon--type
  inequalities and can therefore be verified with the help of any solver
  that handles Shannon--type entropy inequalities.
  
  For all tuples $(X,Y,A,W)$ the following Shannon--type inequality holds:
  \[
    I(X:Y|A)\geq H(W |A) + I(X:Y|A,W)-\big(H(W |X) + H(W |Y)\big), 
  \]
  see inequality~\eqref{eq:shannon-type-strange-2} in the Appendix.
  For the $W$ chosen above, this inequality implies  
  \begin{equation*}
    \begin{split}
      I(X:Y|A)
      &\geq
        H(W |A) + I(X:Y|A,W ) - O(\sqrt{\epsilon}) \cdot I(X:Y)  
       \\
      &\geq
        H(W|A) + I(X:Y|W)  - I(X:Y:A|W) -  O(\sqrt{\epsilon} ) \cdot I(X:Y)\\
       &\geq
        H(W|A) + I(X:Y|W)  - I(X,Y:A|W)  
        -  O(\sqrt{\epsilon} ) \cdot I(X:Y)
     \\
      &\geq
        H(W|A) + I(X:Y|W) -  O(\sqrt{\epsilon} ) \cdot I(X:Y).
    \end{split}
  \end{equation*} 
  It follows that 
  \begin{equation}
    \label{eq:q-2}
    \begin{split}
      H(W |A)
      &\leq
        I(X:Y|A) - I(X:Y|W) +  O(\sqrt{\epsilon})   \cdot I(X:Y)
    \end{split}
  \end{equation}
  Now we adapt the  conventional inference of the Ingleton  inequality. 
  We use the following Shannon--type equalities
  (see~\eqref{eq:common-info-below-mutual}, 
  \eqref{eq:common-info-below-mutual-conditioned}, 
  \eqref{eq:shannon-type-strange-1} in the Appendix):
  \begin{equation}
    \label{eq:inference-ingleton}
    \left.
      \begin{aligned}
        H(W |B)
        &\le
          H(W|X) + H(W|Y) + I(X:Y|B) 
        \\  
        &\leq
          I(X:Y|B)  + O(\sqrt{\epsilon}) \cdot  I(X:Y|A) 
        \\
        &(\text{follows from the choice of $W$)},\\
        H(W)
        &\le
          H(W|A) + H(W|B) + I(A:B),\\ 
        I(X:Y)
        &\leq
          H(W) + I(X:Y|W) 
      \end{aligned}
    \right\}
  \end{equation}
  Summing inequalities~\eqref{eq:q-2} and~\eqref{eq:inference-ingleton} yields the conclusion of the theorem.
\end{proof}

\begin{remark}
  If the quantity $I(X:Y)$ is a priori bounded, e.g., if
  $I(X:Y)\leq1$, then the inequality in Theorem~\ref{th:sqrt} rewrites
  to
  \[
    I(X:Y)
    \le
    I(X:Y|A) + I(X:Y|B) + I(A:B) + O(\sqrt{\epsilon}).
  \]
\end{remark}

Similar techniques can be used to prove Theorem~\ref{th:sqrt-ii}.
\thsqrtii*
\begin{remark}\label{th:sqrt-ii-short}
  Because of the bound $I(X :Y|A) \le \epsilon \cdot I(X :A)$ in the
  assumptions, we can equivalently rewrite the conclusion of 
  Theorem~\ref{th:sqrt-ii} as
\begin{equation*}
I(X :Y) \le I(X :Y|B)+I(A:B)+ O(\sqrt{\epsilon}) \cdot I(X :A).
\end{equation*} 
\end{remark}

\begin{proof}
  Similarly to the previous argument, we apply Corollary~\ref{cor:W-bound}
  (switching $A$ and $Y$ in the assumptions and the conclusion
  of the lemma) and take a $W$ such that
  \begin{itemize}
  \item $H(W|X) = O(\sqrt{\epsilon})\cdot I(X:A)$,
  \item $H(W|A) = O(\sqrt{\epsilon})\cdot I(X:A)$,
  \item $I(X,A:Y|W) = O(\sqrt{\epsilon})\cdot I(X:A)$.
  \end{itemize}
  For brevity, in what follows we use the notation
  $\Delta := \sqrt{\epsilon} \cdot I(X:A)$.  The rest of this proof is
  an application of suitable Shannon--type inequalities.
  
  We begin  with a  Shannon--type inequality
  \begin{equation*}
    H(W) \le H(W|A) + H(W|B) + I(A:B),
  \end{equation*}
  see~\eqref{eq:common-info-below-mutual} in the Appendix. The term
  $H(W|A)$ in the right-hand side of the inequality above is bounded by
  $O(\Delta)$, so we can rewrite this inequality as
  \begin{equation}
  \label{eq:thB-2}
    H(W) \le H(W|B) + I(A:B) + O(\Delta).
  \end{equation}  
  To bound the term $ H(W|B)$, we use another Shannon--type inequality
  \begin{equation*}
    H(W|B) \le H(W|X) + H(W|Y)+ I(X:Y|B),
  \end{equation*} 
  see~\eqref{eq:common-info-below-mutual-conditioned} in the Appendix.
  The term $H(W|X)$ in the right-hand side of the last inequality can be
  replaced immediately by $O(\Delta)$, and we obtain
      \begin{equation}
  \label{eq:thB-4}
    H(W|B) \le  H(W|Y)+ I(X:Y|B) + O(\Delta).
  \end{equation} 
  Our next step is an upper bound for $H(W|Y)$.
  We isolate this step of the proof into a separate claim.
  \begin{claim}
    \(H(W|Y) \le H(W) - I(X:Y)  + O(\Delta).\)
  \end{claim}
  \begin{proof}[Proof of the claim]
    First, applying several times the chain rule for the mutual information, we obtain
    \begin{equation*}
      \begin{split}
        I(X:Y)
        &\le
          I(X:Y)  + I(Y:W|X) \\
        &=
          I(X,W:Y) \\
        &=
          I(Y:W)  + I(X:Y|W) ,\\
        &\le
          I(Y:W)  + I(X:Y|W) + I(Y:A|X,W) \\
        & =
          I(Y:W)  + I(X,A:Y|W) \\
        &=
          I(Y:W) + O(\Delta),
      \end{split}
    \end{equation*}
   where the last step uses the assumption  $I(X,A:Y|W)   = O(\Delta)$. 
    
   It follows that 
   \begin{equation*}
     \begin{split}
       H(W|Y)
       &=
         H(W) - I(Y:W)\\
       &\le
         H(W) - I(X:Y) + O(\Delta)\\
     \end{split}
   \end{equation*}
   and the claim is proven.
 \end{proof}
 
 Combining inequalities~\eqref{eq:thB-2},~\eqref{eq:thB-4}, and the Claim, we obtain
 \begin{equation*}
   H(W) \le H(W) - I(X:Y)   + I(X:Y|B) + I(A:B) + O(\Delta),
 \end{equation*} 
 By Remark~\ref{th:sqrt-ii-short} this completes the proof of the theorem.
\end{proof}

\section{Subtler claims with the $O(\sqrt{\epsilon})$ precision}
\label{s:subtler}
In this section we combine Lemma~\ref{l:Wk-bounds} with symmetry
considerations, derive a new information inequality for six-tuples of
random variables, and prove Theorem~\ref{th:subtler} from the introduction.

\thsubtler*
\begin{remark}
The constant hidden in the $O(\cdot)$-notation does not depend on $t$ nor on the distribution of $(X, Y, A, B, C, D)$.
\end{remark}
\begin{proof}
    We first show that, without loss of generality, we may assume that
    the distribution of the tuple $(X,Y,A,B,C,D)$ is symmetric with
    respect to transposition of $A$ and $B$.  We begin with
    tropicalization: we replace the original tuple by its tropical
    version $(X,Y,A,B,C,D)^{tr}$, where each variable is replaced by
    the joint of $n$ independent copies of itself. The assumptions of
    the theorem still hold for the tropical tuple; see
    Remark~\ref{r:trop2class2trop}.  We then apply
    Lemma~\ref{l:symmetrization} and transform this tropical tuple
    into a symmetric one that still satisfies the assumptions of the
    theorem.  Note that the entropies of the symmetrized tropical
    tuple coincide with those of the original tuple.  Thus, if the
    theorem holds for the symmetrized tuple, then it also holds for
    the original one.  

  Now we are ready to prove the theorem.  To
  address the first claim of the theorem, we apply
  Corollary~\ref{cor:W-bound} to construct $W$ such that
  \begin{itemize}
  \item $H(W|X) = O(\sqrt{\epsilon})\cdot I(X:Y)$,
  \item $H(W|Y) = O(\sqrt{\epsilon})\cdot I(X:Y)$,
  \item $I(X,Y:A|W) = O(\sqrt{\epsilon})\cdot I(X:Y)$.
  \end{itemize}
  Due to Lemma~\ref{l:Wk-bounds} (the \emph{moreover} part) we may
  assume that all entropy quantities for $(X,Y,A,W)$ are the same as
  the corresponding quantities for $(X,Y,B,W)$.  Similarly
  to~\eqref{eq:q-2}, using the symmetry we have
  \begin{align*}
    H(W|A)
    &\le
      I(X:Y|A) - I(X:Y|W)  + O(\sqrt{\epsilon})\cdot I(X:Y)\\ 
    H(W|B)
    &\le
      I(X:Y|B)  -  I(X:Y|W) + O(\sqrt{\epsilon})\cdot I(X:Y)
  \end{align*}
  The sum of two inequalities above together with Shannon--type inequalities
  \begin{align*}
    I(X:Y)
    &\le
      H(W) + I(X:Y|W), \\
    H(W)
    &\le
      H(W|A) + H(W|B) + I(A:B)
  \end{align*}
  gives 
  \begin{equation}
    \label{eq:subtler1}
    \begin{split}
      I(X:Y)
      &\leq
        I(X:Y|A) + I(X:Y|B) + I(A:B)  - I(X:Y|W) \\
      &+
        O(\sqrt{\epsilon})\cdot I(X:Y)
    \end{split}
  \end{equation}
  Therefore, if $t\le I(X:Y|W)/I(X:Y)$, then 
  \[
    I(X:Y) \le I(X:Y|A) + I(X:Y|B) + I(A:B) +
    \big( O(\sqrt{\epsilon})-t)\cdot I(X:Y).
  \]
  This gives the first part of the alternative in
  claim~\ref{th:subtler-a}.
    
  On the other hand, for all $(X,Y,C,D,W)$ we can follow the usual
  scheme of the proof of the Ingleton's inequality with $W$
  representing the common information of $X$ and $Y$. More formally,
  we add Shannon--type inequalities
  \begin{align*}
    H(W) &\leq H(W|C) + H(W|D) + I(C:D),
    \\
    H(W|C) &\leq H(W|X) + H(W|Y) + I(X:Y|C),
    \\
    H(W|D) &\leq H(W|X) + H(W|Y) + I(X:Y|D),
    \\
    I(X:Y)  &\leq H(W) + I(X:Y|W)
  \end{align*}
  to obtain
  \begin{align*}
    I(X:Y)
    &\leq
      I(X:Y|C) + I(X:Y|D) + I(C:D) \\
    &+
      2\big(H(W|X) + H(W|Y)\big)
      +I(X:Y|W)
  \end{align*}
  Therefore, for the chosen $W$ we have 
  \begin{equation}
    \label{eq:subtler2}
    \begin{split}
      I(X:Y)
      &\leq
        I(X:Y|C) + I(X:Y|D) + I(C:D) \\
      &+
        O(\sqrt{\epsilon})\cdot I(X:Y)+I(X:Y|W)
    \end{split}
  \end{equation}
  Thus, if $t\geq I(X:Y|W)/I(X:Y)$, we obtain 
  \[
    I(X:Y) \leq
    I(X:Y|C) + I(X:Y|D) + I(C:D) +
    \big(O(\sqrt{\epsilon})+t\big)\cdot I(X:Y).
  \]
  This gives the second part of the alternative.

  To prove part~\ref{th:subtler-b} of the theorem, we take the sum of
  inequalities~\eqref{eq:subtler1} and~\eqref{eq:subtler2}.
\end{proof}

\section{Inequalities with the precision of
  $O(\epsilon\log\frac1\epsilon)$}
\label{s:eloge}
This section is devoted to the proof of Theorem~\ref{th:eloge}. For
convenience the statement is repeated below verbatim.

\theloge*

\begin{proof}
  The conclusion of the theorem (the Ingleton inequality) is trivial in case $I(X:Y) = 0$. So, in what follows we assume without loss of generality that the mutual information $I(X:Y)$ is strictly positive.

  Since the inequalities in both the assumption and the conclusion of
  the theorem are 
  symmetric with respect to swapping $A$ and $B$,
  we may apply
  Lemma~\ref{l:symmetrization} to the four-tuple $(X,Y,A,B)$. Thus,
  without loss of generality, we may assume that the distribution is
  symmetric  (see also Remark~\ref{r:trop2class2trop}).
  
  Now we use Lemma~\ref{l:Wk-bounds} and symmetry of the distribution
  to find a sequence of extensions $(X,Y,A,B,W_{i})$, $i=0,1,\dots$
  satisfying
  \begin{align*}
    H(W_{i}|X,Y) &= 0\\
    I(X:Y|W_{i+1}) &= I(X:Y|W_{i},A)\\
    I(X:Y|W_{i+1}) &= I(X:Y|W_{i},B)\\
    \max\set{H(W_{i}|X),H(W_{i}|Y)} &\leq i\cdot\epsilon\cdot I(X:Y)
  \end{align*}
  Fix two numbers $\lambda,\delta\in (0,1)$. We postpone the concrete
  choices for these two parameters for later in the proof. Take the
  largest $k$ such that for all $0\leq i < k$ holds
  \begin{align}
    \label{eq:eloge-1}
    I(X:Y|W_{i+1})
    &\leq
      \lambda\cdot I(X:Y| W_{i} )\quad\text{and}
    \\
    \label{eq:eloge-2}
    I(X:Y| W_{i})
    &\geq
      \delta\cdot I(X:Y)
  \end{align}
  Then
  \begin{equation*}
    \delta\cdot I(X:Y)\leq I(X:Y| W_{k} )
    \leq \lambda I(X:Y| W_{k-1} ) \leq \ldots
    \leq\lambda^{k}\cdot I(X:Y).
  \end{equation*}
  Since we assume that $I(X:Y)>0$, we may conclude that 
  \begin{equation}\label{eq:k0-bound}
    k\leq\frac{\ln\delta}{\ln\lambda}
  \end{equation}
  
  There are two possible cases, depending on whether
  condition~\eqref{eq:eloge-1} or~\eqref{eq:eloge-2} is violated at $i=k$.
  \subsection{Case 1: $I(X:Y|W_{k+1} )\geq\lambda\cdot I(X:Y| W_{k} )$}
  We use inequality~\eqref{eq:shannon-type-strange-2} from the Appendix
  to obtain 
  \begin{equation*}
    \begin{split}
      I(X:Y|A)
      &\geq
        H(W_{k}|A) + I(X:Y|W_{k} , A) - \big(H(W_{k}|X,A) + H(W_{k}|Y,A)\big)\\
      &\geq
        H(W_{k}|A) + I(X:Y|W_{k+1}) - \big(H(W_{k}|X) + H(W_{k}|Y)\big)\\
      &\geq
        H(W_{k}|A) + \lambda\cdot I(X:Y|W_{k}) - 2k\cdot\epsilon\cdot I(X:Y)
    \end{split}
  \end{equation*}
  Similarly we have
  \[
    I(X:Y|B)\geq
    H(W_{k}|B) + \lambda\cdot I(X:Y|W_{k}) -
    2k\cdot\epsilon\cdot I(X:Y)
  \]
  We can rewrite these inequalities as
  \begin{align}
    \label{eq:case11}
    H(W_{k}|A)
    &\leq
      I(X:Y|A) - \lambda\cdot I(X:Y|W_{k}) + 2k\cdot\epsilon\cdot I(X:Y)\\
    H(W_{k}|B)
    &\leq
      I(X:Y|B) - \lambda\cdot I(X:Y|W_{k}) + 2k\cdot\epsilon\cdot
      I(X:Y)
  \end{align}
  Applying two Shannon inequalities~\eqref{eq:common-info-below-mutual}
  and~\eqref{eq:shannon-type-strange-1} from the Appendix we get  
  \begin{align}
    H(W_{k})
    &\leq
      H(W_{k}|A) + H(W_{k}|B) + I(A:B)\\
    \label{eq:case14}
    I(X:Y)
    &\leq
      H(W_{k}) + I(X:Y|W_{k})
  \end{align}
  Adding inequalities~(\ref{eq:case11}--\ref{eq:case14}) we get
  \begin{equation}
    \begin{split}
      I(X:Y)
      &\leq
        I(X:Y|A) + I(X:Y|B) + I(A:B)\\
      &\quad+
        (1-2\lambda)I(X:Y|W_{k}) + 4k\cdot\epsilon\cdot I(X:Y)
    \end{split}
  \end{equation}
  
  Choosing $\lambda=1/2$ and using the bound~\eqref{eq:k0-bound} for $k$ we obtain
  \begin{equation}\label{eq:eloge-bound1}
    \begin{split}
      I(X:Y)
      &\leq
        I(X:Y|A) + I(X:Y|B) + I(A:B)\\
      &\quad+
        4\epsilon\log_{2}\delta^{-1}\cdot I(X:Y)    
    \end{split}
  \end{equation}
  
  \subsection{Case 2: $I(X:Y| W_{k} )\leq\delta\cdot I(X:Y)$}
  We apply inequality~\eqref{eq:mmrv} from the Appendix, and use the
  properties of $W_i$ from Lemma~\ref{l:Wk-bounds} to estimate
  \begin{align*}
    I(X:Y)
    &\leq
      I(X:Y|A) + I(X:Y|B) + I(A:B) \\
    &\quad+
      I(X:Y| W_{k} ) + I(X: W_{k} |Y) + I(Y:W_{k}|X)\\
    &\leq
      I(X:Y|A) + I(X:Y|B) + I(A:B) + (\delta + 2 k\cdot\epsilon)
      I(X:Y)
  \end{align*}
  Using~\eqref{eq:k0-bound} we rewrite the last inequality as 
  \begin{equation}
    \label{eq:eloge-bound2}
    \begin{split}
      I(X:Y)
      &\leq
        I(X:Y|A) + I(X:Y|B) + I(A:B)\\
      &\quad+ (\delta + 2 \epsilon\log_{2}\delta^{-1})I(X:Y)
    \end{split}
  \end{equation}
  
  Now we set $\delta=\epsilon$ in inequalities~\eqref{eq:eloge-bound1}
  and~\eqref{eq:eloge-bound2} and take the maximum of the two
  bounds. This gives
  \begin{equation*}
    I(X:Y)\leq I(X:Y|A) + I(X:Y|B) + I(A:B) + O(\epsilon\log\epsilon^{-1})I(X:Y)
  \end{equation*}
\end{proof}

\appendix
\section*{Appendix: Useful information inequalities}
In this section we collect several equalities and inequalities for the
Shannon entropy that are used throughout the paper.  First of all, we
freely use the chain rule for the mutual information
\[
I(X,Y:A) = I(X:A) + I(Y:A|X)
\]
and its ``conditioned'' version 
\[
I(X,Y:A|B) = I(X:A|B) + I(Y:A|X,B).
\]
Observe that it implies  
\begin{equation}
\label{eq:chain-rule-simplified}
I(X,Y:A|B) \ge I(Y:A|X,B).
\end{equation}

Besides, we apply several times the well-known Shannon--type inequality 
\begin{equation}
  \label{eq:common-info-below-mutual}
  H(Z) \le H(Z|X) + H(Z|Y) + I(X:Y)
\end{equation}
(which is equivalent to the sum of $I(X:Y|Z)\ge 0$ and $H(Z|X,Y)\ge 0$)
and its conditioned version
\begin{equation}
  \label{eq:common-info-below-mutual-conditioned}
  H(Z|W) \le H(Z|X,W) + H(Z|Y,W) + I(X:Y|W),
\end{equation}
In addition, we use several other Shannon--type inequalities. First of all, we use
\begin{equation}
  \label{eq:shannon-type-strange-1}
  I(X : Y)\le H(W)+I(X :Y|W)
\end{equation}
(which is equivalent to the sum of $H(W|X) \ge 0$ and $I(X:W|Y) \ge 0$)
and its conditional version
\begin{equation}
  \label{eq:shannon-type-strange-1-cond}
  I(X : Y|Z)\le H(W|Z)+I(X :Y|W,Z).
\end{equation}
Besides, we need the inequality 
\begin{equation*}
  H(W) + I(X:Y|W) 
  \leq
  I(X:Y) + \big(H(W|X) + H(W|Y)\big)
\end{equation*}
(which easily rewrites to the trivial $H(W|X,Y)\ge0$), and its
conditioned version
\begin{equation}
  \label{eq:shannon-type-strange-2}
  H(W|A) + I(X:Y|W, A) 
  \leq
  I(X:Y|A) + \big(H(W|X,A) + H(W|Y,A)\big).
\end{equation}
Finally, we use a version of the classical non-Shannon--type inequality
from~\cite{zhang1998characterization} in the form with five variables
(proven in~\cite[theorem~1 for $n=2$]{makarychev2002new}),
\begin{equation}
\nonumber
H(A,B) + 2 I(X:Y:W) \le I(X:Y|A) + I(X:Y|B) + H(A) + H(B) + I(XY:W).
\end{equation}
This inequality is equivalent to
\begin{equation}
\nonumber
  \begin{split}
 2 I(X:Y:W) &\le I(X:Y|A) + I(X:Y|B) + I(A:B) \\
 &\quad + \left(I(X:Y:W)+ I(X:W|Y) + I(Y:W|X)\right).
   \end{split}
\end{equation}
Since $I(X:Y) = I(X:Y:W) + I(X:Y|W)$, it can be rewritten in the form convenient for our application:
\begin{equation}
  \label{eq:mmrv}
  \begin{split}
    I(X:Y)
    &\le
      I(X:Y|A) + I(X:Y|B) + I(A:B) \\
    &\quad+
      I(X:Y|W) + I(X:W|Y) + I(Y:W|X).
  \end{split}
\end{equation}


\bibliographystyle{alpha}   
\bibliography{ingleton-inequality}

\end{document}